\documentclass[11pt]{article}

\usepackage{algorithm}
\usepackage{algorithmic}
\usepackage{subfigure}
\usepackage{epsfig,amsthm,amsmath,color, amsfonts}
\usepackage{epsfig,color}
\usepackage{framed}
\usepackage{psfrag}
\usepackage{hyperref}
\usepackage{fullpage}
\usepackage{color}
\usepackage{tikz}
\usetikzlibrary{shapes,arrows}

\newtheorem{theorem}{Theorem}[section]

\newtheorem{lemma}[theorem]{Lemma}
\newtheorem{claim}[theorem]{Claim}

\theoremstyle{definition}\newtheorem{example}[theorem]{Example}
\theoremstyle{definition}
\theoremstyle{invariant}

\newcommand{\comment}[1]{}
\newcommand{\QED}{\mbox{}\hfill \rule{3pt}{8pt}\vspace{10pt}\par}

\def\polylog{\operatorname{polylog}}

\newcommand{\ignore}[1]{}
\newcommand{\eat}[1]{}

\newcommand{\diam}{\Lambda}

\newcommand{\graph}{G(\Gamma, \kappa, \Lambda)}

\newcommand{\squishlist}{\begin{itemize}}
\newcommand{\squishend}{\end{itemize}}

\def\PC{\mbox{\sc pc}}

\def\cP{\mathcal{P}}
\def\cH{\mathcal{H}}
\def\cA{\mathcal{A}}

\def\cF{\mathcal{F}}

\def\danupon#1{}
\def\atish#1{}
\def\gopal#1{}
\def\note#1{}

\begin{document}

\title{A tight unconditional lower bound on distributed random walk computation}
\author{
Danupon Nanongkai\thanks{The University of Vienna, Vienna, Austria. \hbox{E-mail}:~\url{danupon@cc.gatech.edu}. Part of this work done while at Georgia Institute of Technology.}
\and
Atish {Das Sarma}\thanks{Google Research, Google Inc., Mountain View, USA. \hbox{E-mail}:~\url{dassarma@google.com}.}
\and
Gopal Pandurangan\thanks{Division of Mathematical
Sciences, Nanyang Technological University, Singapore 637371 and Department of Computer Science, Brown University, Providence, RI 02912, USA.  \hbox{E-mail}:~\url{gopalpandurangan@gmail.com}. Supported in part by NSF grant CCF-1023166 and by a grant from the United States-Israel Binational Science
Foundation (BSF).}
}

\maketitle

\begin{abstract}
We consider the problem of performing a random walk in a distributed network. Given bandwidth constraints, the goal of the problem is to minimize the number of rounds required to obtain a random walk sample. Das Sarma et al. [PODC'10] show that a random walk of length $\ell$ on a network of diameter $D$ can be performed in $\tilde O(\sqrt{\ell D}+D)$ time. A major question left open is whether there exists a faster algorithm, especially whether the multiplication of $\sqrt{\ell}$ and $\sqrt{D}$ is necessary.

In this paper, we show a tight unconditional lower bound on the time complexity of distributed random walk computation. Specifically, we show that for any $n$, $D$, and $D\leq \ell \leq (n/(D^3\log n))^{1/4}$, performing a random walk of length $\Theta(\ell)$ on an $n$-node network of diameter $D$ requires $\Omega(\sqrt{\ell D}+D)$ time. This bound is {\em unconditional}, i.e., it holds for any (possibly randomized) algorithm. To the best of our knowledge, this is the first lower bound that the diameter plays a role of multiplicative factor. Our bound shows that the algorithm of Das Sarma et al. is time optimal.

Our proof technique introduces a new connection between {\em bounded-round} communication complexity and distributed algorithm lower bounds with $D$ as a trade-off parameter, strengthening the previous study by Das Sarma et al. [STOC'11]. In particular, we make use of the bounded-round communication complexity of the pointer chasing problem. Our technique can be of independent interest and may be useful in showing non-trivial lower bounds on the complexity of other fundamental distributed computing problems.
\end{abstract}

\noindent {\bf Keywords:} Random walks, Distributed algorithms, Lower bounds, Communication complexity. \\




\section{Introduction}


The random walk plays a central role in computer science, spanning a wide range of areas in both theory and practice. The focus of this paper is on performing a random walk in distributed networks, in particular, decentralized algorithms for performing a random walk in arbitrary networks.
The random walk is used as an integral subroutine in a wide variety of network applications ranging from token management~\cite{IJ90,BBF04, CTW93},
small-world routing~\cite{K00}, search~\cite{ZS06,AHLP01,C05,GMS05,LCCLS02}, information propagation and gathering~\cite{BAS04,KKD01}, network topology construction~\cite{GMS05,LawS03,LKRG03}, expander testing~\cite{DolevT10}, constructing random spanning
trees~\cite{Broder89, BIZ89, BFG+03},
distributed construction of expander networks \cite{LawS03}, and peer-to-peer membership management~\cite{GKM03,ZSS05}.
For more applications of random walks to distributed networks, see e.g.~\cite{DasSarmaNPT-PODC10}.
%
%
Motivated by the wide applicability of the random walk, \cite{DasSarmaNP09,DasSarmaNPT-PODC10} consider the running time of performing a random walk on the synchronous distributed model. We now explain the model and problems before describing previous work and our results.


\subsection{Distributed Computing Model}
Consider a synchronous  network of processors with unbounded computational power. The network is modeled by an undirected connected $n$-node multi-graph, where nodes model the processors and  edges model the links between the processors. The processors  (henceforth, nodes) communicate  by exchanging messages via the links (henceforth, edges).  The nodes  have limited global knowledge, in particular, each of them has its own local perspective of the network (a.k.a graph), which is confined to its immediate neighborhood.

There are several measures to analyze the performance of algorithms on this model, a fundamental one being the running time, defined as the worst-case number of {\em rounds} of distributed communication. This measure naturally gives rise to a  complexity measure of problems, called the {\em time complexity}. On each round at most $O(\log n)$ bits can be sent through each edge in each direction.
This is a standard model of distributed computation known as the ${\cal CONGEST}$ model~\cite{peleg} and has been attracting a lot of research attention during last two decades (e.g., see \cite{peleg} and the references therein). We note that our result also holds on the ${\cal CONGEST}(B)$ model, where on each round at most $B$ bits can be sent through each edge in each direction (see the remark after Theorem~\ref{thm:rw_lower_bound}). We ignore this parameter to make the proofs and theorem statements simpler.

\subsection{Problems}
The basic problem is {\em computing a random walk where destination outputs source}, defined as follows.
We are given a network $G = (V,E)$ and a source node $s \in V$. The goal is to devise a distributed algorithm such that, in the end, some node $v$ outputs the ID of $s$, where $v$ is a destination node picked according to the probability that it is the destination of a random walk of length $\ell$ starting at $s$. 
We assume the standard random walk where, in each step, an edge is taken from the current node $v$ with probability proportional to $1/d(v)$ where $d(v)$ is the degree of $v$. Our goal is to output a true  random sample from the $\ell$-walk distribution starting from $s$.

For clarity, observe that the following naive algorithm solves the above problem in $O(\ell)$ rounds. The walk of length $\ell$ is performed by sending a token for $\ell$ steps, picking a random neighbor in each step. Then, the destination node $v$ of this walk outputs the ID of $s$.
The main objective of distributed random walk problem is to perform such sampling with significantly less number of rounds, i.e., in time that is sublinear in $\ell$.  On the other hand, we note that it can take too much time (as much as $\Theta(|E|+D)$ time) in the ${\cal CONGEST}$  model to collect all the topological information at the source node (and then computing the walk locally).

The following variations were also previously considered.
\begin{enumerate}
\item \textit{Computing a random walk where source outputs destination}: The problem is almost the same as above except that, in the end, the source has to output the ID of the destination. This version is useful in nodes learning the topology of their surrounding networks and related applications such as a decentralized algorithm for estimating the mixing time \cite{DasSarmaNPT-PODC10}.

\item \textit{Computing a random walk where nodes know their positions}: Instead of outputting the ID of source or destination, we want each node to know its position(s) in the random walk. That is, if $v_1, v_2, ..., v_\ell$ (where $v_1=s$) is the result random walk starting at $s$, we want each node $v_j$ in the walk to know the number $j$ at the end of the process. This version is used to construct a random spanning tree in \cite{DasSarmaNPT-PODC10}.
\end{enumerate}


\subsection{Previous work and our result}

A key purpose of the random walk in many  network applications is to perform  node sampling.  While the sampling requirements in different applications vary, whenever a true sample is required from a random walk of certain steps, typically all applications perform the walk naively --- by simply passing a token from one node to its neighbor: thus to perform a random walk of length $\ell$ takes time linear in $\ell$.
%
Das Sarma et al.~\cite{DasSarmaNP09} showed this is not a time-optimal strategy and the running time can be made sublinear in $\ell$, i.e., performing a random walk of length $\ell$ on an $n$-node network of diameter $D$ can be done in $\tilde O(\ell^{2/3}D^{1/3})$ time where $\tilde O$ hides $\polylog n$. Subsequently, Das Sarma et al.~\cite{DasSarmaNPT-PODC10} improved this bound to $\tilde O(\sqrt{\ell D}+D)$ which holds for all three versions of the problem.

There are two key motivations for obtaining sublinear time bounds. The first is that in many algorithmic applications, walks of length significantly greater than the network diameter are needed.  For example, this is necessary in two applications presented in \cite{DasSarmaNPT-PODC10}, namely distributed computation of a random
spanning tree (RST) and computation of mixing time. More generally, many real-world communication networks  (e.g., ad hoc networks and peer-to-peer networks) have relatively small diameter, and random walks of length at least the diameter are usually performed for many sampling applications, i.e., $\ell >> D$.

The second motivation is understanding the time complexity of distributed random walks. Random walk is essentially a global problem  which requires the algorithm to ``traverse" the entire network. Classical ``global" problems include the minimum spanning tree, shortest path etc. Network diameter is an inherent lower bound for such problems. Problems of this type raise the basic question whether $n$ (or $\ell$ as the case here) time is essential or is the network diameter $D$, the inherent parameter. As pointed out in the seminal work of \cite{peleg-mst}, in the latter case, it would be desirable to design algorithms that have a better complexity for graphs with low diameter. While both upper and lower bounds of time complexity of many ``global'' problems are known (see, e.g., \cite{DasSarmaHKKNPPW10}), the status of the random walk problem is still wide open.


A preliminary attempt to show a random walk lower bound is presented in \cite{DasSarmaNPT-PODC10}. They consider a restricted class of algorithms, where each message sent between nodes must be in the form of an interval of numbers. Moreover, a node is allowed to send a number or an interval containing it only after it receives such number. For this very restricted class of algorithms, a lower bound of $\Omega(\sqrt{\ell}+D)$ is shown~\cite{DasSarmaNPT-PODC10} for the version where every node must know their positions in the end of the computation.
While this lower bound shows a potential limitation of random walk algorithms,
it has many weaknesses. First, it does not employ an information theoretic argument and thus does not hold for all types of algorithm. Instead, it assumes that the algorithms must send messages as intervals and thus holds only for a small class of algorithms, which does not even cover all deterministic algorithms.
Second, the lower bound holds only for the version where nodes know their position(s), thereby leaving lower bounds for the other two random walk versions completely open.
More importantly, there is still a gap of $\sqrt{D}$ between lower and upper bounds, leaving a question whether there is a faster algorithm.


Motivated by these applications, past results, and open problems, we consider the problem of finding lower bounds for random walk computation. In this work, we show an {\em unconditional} lower bound of $\Omega(\sqrt{\ell D}+D)$ for all three versions of the random walk computation problem. This means that the algorithm in \cite{DasSarmaNPT-PODC10} is optimal for all three variations. In particular, we show the following theorem.



\begin{theorem}\label{thm:rw_lower_bound}
For any $n$, $D$ and $\ell$ such that $D\leq \ell\leq (n/(D^3\log n))^{1/4}$, there exists a family of $n$-node networks of diameter $D$ such that performing a random walk (any of the three versions) of length $\Theta(\ell)$ on these networks requires $\Omega(\sqrt{\ell D}+D)$ rounds.
\end{theorem}

%

%
%
We note that our lower bound of $\Omega(\sqrt{\ell D}+D)$ also holds for the general ${\cal CONGEST}(B)$ model, where each edge has bandwidth $B$ instead of $O(\log n)$, as long as $\ell\leq (n/(D^3 B))^{1/4}$.
Moreover, one can also show a lower bound on simple graphs by subdividing edges in the network used in the proof and double the value of $\ell$.

\subsection{Techniques and proof overview}

Our main approach relies on enhancing the connection between communication complexity and distributed algorithm lower bounds first studied in \cite{DasSarmaHKKNPPW10}. It has been shown in \cite{DasSarmaHKKNPPW10} that a fast distributed algorithm for computing a function can be converted into a two-party communication protocol with small message complexity to compute the same function. In other words, the communication complexity lower bounds implies the time complexity lower bounds of distributed algorithms. This result is then used to prove lower bounds on many {\em verification problems}. (In the verification problems, we are given $H$, a subgraph of the network  $G$, where each vertex of $G$ knows which edges incident on it are in $H$. We would like to verify whether $H$ has some properties, e.g., if it is a tree or if it is connected.) The lower bounds of verification problems are then used to prove lower bounds of approximation algorithms for many graph problems. Their work, however, does not make progress on achieving {\em any} unconditional lower bound on the random walk problem.

Further, while this approach has been successfully used to show lower bounds for several problems in terms of network size (i.e., $n$), it is not clear how to apply them to random walk computation. All the lower bounds previously shown are for {\em optimization} problems for well-defined metrics - for e.g. computing a minimum spanning tree. Random walk computation, on the other hand, is not deterministic; the input requires parameters such as the length of the walk $\ell$ and even if the source node and $\ell$ are fixed, the {\em solution} (i.e. the walk) is not uniquely determined. While even other problems, such as MST, can have multiple solutions, for optimization problems, {\em verification} is well-defined. It is not clear what it even means to verify whether a random walk is {\em correct}. For this reason, proving a lower bound of random walk computation through verification problems seems impossible.

Additionally, in terms of the theoretical bound we obtain, a key difficulty in our result is to introduce the graph parameter diameter (i.e., $D$) into the lower bound multiplied by $\ell$. A crucial shortcoming in extending previous work in this regard is that the relationship between communication complexity and distributed computing shown in \cite{DasSarmaHKKNPPW10} does not depend on the network diameter $D$ at all. In fact, such a relationship might not exist since the result in \cite{DasSarmaHKKNPPW10} is tight for some functions.

To overcome these obstacles, we consider a variation of communication complexity called {\em  $r$-round two-party communication complexity}, which has been successfully used in, e.g., circuit complexity and data stream computation (see, e.g., \cite{FeigenbaumKMSZ08,NisanW93}). We obtain a new connection showing that a fast distributed algorithm for computing a function can be converted to a two-party communication protocol with a small message complexity {\em and number of rounds} to compute the same function. Moreover, the larger the network diameter is, the smaller the number of rounds will be.
To obtain this result one need to deal with a more involved proof; for example, the new proof does not seem to work for the networks previously considered \cite{DasSarmaHKKNPPW10,Elkin06,KorKP11,LotkerPP06,PelegR00} and thus we need to introduce a new network called $\graph$ (which is essentially an extension of the network $F^2_m$ in \cite{PelegR00}). This result and related definitions are stated and proved in Section~\ref{sec:communication_complexity}.

A particular communication complexity result that we will use is that of Nisan and Wigderson~\cite{NisanW93} for the {\em $r$-round pointer chasing problem}. Using the connection established in Section~\ref{sec:communication_complexity}, we derive a lower bound of any distributed algorithms for solving the pointer chasing problem on a distributed network. This result is in Section~\ref{sec:pointer_chasing}.

Finally, we prove Theorem~\ref{thm:rw_lower_bound} from the lower bound result in Section~\ref{sec:pointer_chasing}. The main idea, which was also used in \cite{DasSarmaNPT-PODC10}, is to construct a network that has the same structure as $\graph$ (thus has the same diameter and number of nodes) but different edge capacities (depending on the input) so that a random walk follows a desired path (which is unknown) with high probability.
This proof is in Section~\ref{sec:main_theorem}.

\section{From bounded-round communication complexity to dis\-tri\-bu\-ted algorithm lower bounds}\label{sec:communication_complexity}


Consider the following problem. There are two parties that have unbounded computational power. Each party receives a $b$-bit string, for some integer $b\geq 1$, denoted by $\bar{x}$ and $\bar{y}$ in $\{0, 1\}^b$. They both want to together compute $f(\bar{x}, \bar{y})$ for some function $f:\{0, 1\}^b\times \{0, 1\}^b\rightarrow \mathbb{R}$. At the end of the computation, the party receiving $\bar{y}$ has to output the value of $f(\bar{x}, \bar{y})$. We consider two models of communication.

\squishlist
\item {\em $r$-round direct communication:} This is a variant of the standard model in communication complexity (see \cite{NisanW93} and references therein). Two parties can communicate via a bidirectional edge of unlimited bandwidth. We call the party receiving $\bar{x}$ {\em Alice}, and the other party {\em Bob}. Two parties communicate in {\em rounds} where each round Alice sends a message (of any size) to Bob followed by Bob sending a message to Alice.

\item {\em Communication through network $\graph$:} Two parties are distinct nodes in a distributed network $\graph$, for some integers $\Gamma$ and $\Lambda$ and real $\kappa$; all networks in $\graph$ have $\Theta(\kappa\Gamma\Lambda^\kappa)$ nodes and a diameter of $\Theta(\kappa\Lambda)$.
    (This network is described below.) We denote the nodes receiving $\bar{x}$ and $\bar{y}$ by $s$ and $t$, respectively. 
\squishend

We consider the {\em public coin randomized algorithms} under both models. In particular, we assume that all parties (Alice and Bob in the first model and all nodes in $\graph$ in the second model) share a random bit string of infinite length.
For any $\epsilon\geq 0$, we say that a randomized algorithm $\mathcal{A}$ is {\em $\epsilon$-error} if for any input, it outputs the correct answer with probability at least $1-\epsilon$, where the probability is over all possible random bit strings.
In the first model, we focus on the message complexity, i.e., the total number of bits exchanged between Alice and Bob, denoted by  $R_{\epsilon}^{r-cc-pub}(f)$. In the second model, we focus on the running time, denoted by $R_\epsilon^{\graph, s, t}(f)$.

Before we describe $\graph$ in detail, we note the following characteristics which will be used in later sections. An essential part of $\graph$ consists of $\Gamma$ {\em paths}, denoted by $\cP^1, \ldots, \cP^\Gamma$ and nodes $s$ and $t$ (see Fig.~\ref{fig:graph}). Every edge induced by this subgraph has infinitely many copies (in other words, infinite capacity). (We let some edges to have infinitely many copies so that we will have a freedom to specify the number of copies later on when we prove Theorem~\ref{thm:rw_lower_bound} in Section~\ref{sec:main_theorem}.
%
The leftmost and rightmost nodes of each path are adjacent to $s$ and $t$ respectively. Ending nodes on the same side of the path (i.e., leftmost or rightmost nodes) are adjacent to each other.
%
%
The following properties of $\graph$ follow from the construction of $\graph$ described in Section~\ref{subsec:graph_description}. 
\begin{lemma}\label{lem:graphsize} For any $\Gamma\geq 1$, $\kappa\geq 1$ and $\Lambda\geq 2$, network $\graph$ has $\Theta(\Gamma\kappa\Lambda^\kappa)$ nodes. Each of its path $\cP^i$ has $\Theta(\kappa\Lambda^{\kappa})$ nodes. Its diameter is $\Theta(\kappa\diam)$.
\end{lemma}
\begin{proof} 
It follows from the construction of $\graph$ in Section~\ref{subsec:graph_description} that the number of nodes in each path $\cP^i$ is
$\sum_{j=-\lceil\kappa\rceil\Lambda^{\lfloor\kappa\rfloor}}^{\lceil\kappa\rceil\Lambda^{\lfloor\kappa\rfloor}} \phi'_j = \Theta(\kappa\Lambda^\kappa)$ (cf. Eq.~\eqref{eq:sum_phi'}).
\note{The real value is between $2\Lambda^\kappa$ and $2(\Lambda^\kappa+\Lambda^{\lfloor\kappa\rfloor})$}
Since there are $\Gamma$ paths, the number of nodes in all paths is $\Theta(\Gamma\kappa\Lambda^\kappa)$.
Each highway $\cH^i$ has $2\lceil\kappa\rceil\Lambda^i+1$ nodes. Therefore, there are $\sum_{i=1}^{\lfloor\kappa\rfloor} (2\lceil\kappa\rceil\Lambda^i+1)$ nodes in the highways. For $\Lambda\geq 2$, the last quantity is $\Theta(\lceil\kappa\rceil\Lambda^{\lfloor\kappa\rfloor})$.
Hence, the total number of nodes is $\Theta(\Gamma\kappa\Lambda^\kappa)$.

To analyze the diameter of $\graph$, observe that each node on any path $\cP^i$ can reach a node in highway $\cH^{\lfloor\kappa\rfloor}$ by traveling through $O(\kappa\Lambda)$ nodes in $\cP^i$. Moreover, any node in highway $\cH^i$ can reach a node in highway $\cH^{i-1}$ by traveling trough $O(\Lambda)$ nodes in $\cH^i$. Finally, there are $O(\kappa\Lambda)$ nodes in $\cH^1$. Therefore, every node can reach any other node in $O(\kappa\Lambda)$ steps by traveling through $\cH^1$. Note that this upper bound is tight since the distance between $s$ and $t$ is $\Omega(\kappa\Lambda)$.
\end{proof}

The rest of this section is devoted to prove Theorem~\ref{thm:cc_to_distributed} which strengthens Theorem~3.1 in \cite{DasSarmaHKKNPPW10}. Recall that Theorem~3.1 in \cite{DasSarmaHKKNPPW10} states that if there is a fast $\epsilon$-error algorithm for computing function $f$ on any network $\graph$, then there is a fast $\epsilon$-error algorithm for Alice and Bob to compute $f$, as follows\footnote{Note that Theorem~3.1 in \cite{DasSarmaHKKNPPW10} is in fact stated on a graph different from $\graph$ but its proof can be easily adapted to prove Theorem~\ref{thm:prev_cc_to_distributed}.}.

\begin{theorem}[Theorem~3.1 in \cite{DasSarmaHKKNPPW10}]\label{thm:prev_cc_to_distributed}
Consider any integers $\Gamma\geq 1$, $\Lambda\geq 2$, real $\kappa\geq 1$ and function $f:\{0, 1\}^{b}\times \{0, 1\}^{b} \rightarrow \mathbb{R}$. Let $r=R_\epsilon^{\graph, s, t}(f)$. For any $b$, if $r \le \kappa\Lambda^\kappa$ then $f$ can be computed by a direct communication protocol using at most $(2\kappa\log{n}) r$ communication bits in total. In other words,
$$R_{\epsilon}^{\infty-cc-pub}(f)\leq (2\kappa\log{n})R_\epsilon^{\graph, s, t}(f)\,.$$
\end{theorem}

The theorem above does not try to optimize number of rounds used by direct communication protocols. In fact, a closer look into the proof of Theorem~3.1 in \cite{DasSarmaHKKNPPW10} reveals that $\tilde \Theta((2\kappa\log{n})R_\epsilon^{\graph, s, t}(f))$ rounds of communication are used.

Theorem~\ref{thm:cc_to_distributed} stated below strengthens the above theorem by making sure that the number of rounds needed in the direct communication is small. In particular, it says that if there is a fast $\epsilon$-error algorithm for computing function $f$ on any network $\graph$, then there is a fast {\em bounded-round} $\epsilon$-error algorithm for Alice and Bob to compute $f$. More importantly, the number of rounds depends on the diameter of $\graph$ (which is $\Theta(\kappa\diam)$), i.e., the larger the network diameter, the smaller the number of rounds.


\begin{theorem}\label{thm:cc_to_distributed}
Consider any integers $\Gamma\geq 1$, $\Lambda\geq 2$, real $\kappa\geq 1$ and function $f:\{0, 1\}^{b}\times \{0, 1\}^{b} \rightarrow \mathbb{R}$. Let $r=R_\epsilon^{\graph, s, t}(f)$. For any $b$, if $r \le \kappa\Lambda^\kappa$ then $f$ can be computed by a $\frac{8r}{\kappa\Lambda}$-round direct communication protocol using at most $(2\kappa\log{n}) r$ communication bits in total. In other words,
$$R_{\epsilon}^{\frac{8R_\epsilon^{\graph, s, t}(f)}{\kappa\Lambda}-cc-pub}(f)\leq (2\kappa\log{n})R_\epsilon^{\graph, s, t}(f)\,.$$
%
\end{theorem}

\newcommand{\PRgraph}{F(\Gamma, \kappa, \Lambda)}
\subsection{Preliminary: the network $\PRgraph$}


\begin{figure*}[t]
  \centering
  \tiny
    {
    \psfrag{A}[c]{$\cH^1$}
    \psfrag{B}[c]{$\cH^2$}
    \psfrag{C}[c]{$\cP^1$}
    \psfrag{D}[c]{$\cP^2$}
    \psfrag{E}[c]{$\cP^\Gamma$}
    \psfrag{F}{$v^1_{-12, 1}$}
    \psfrag{G}{$v^1_{-12, 2}$}
    \psfrag{H}[r]{$v^1_{-\infty}$}
    \psfrag{I}{$v^2_{-12, 1}$}
    \psfrag{J}{$v^2_{-12, 2}$}
    \psfrag{K}{$v^2_{-\infty}$}
    \psfrag{L}{$v^\Gamma_{-12, 1}$}
    \psfrag{M}{$v^\Gamma_{-12, 2}$}
    \psfrag{N}[r]{$v^\Gamma_{-\infty}$}
    \psfrag{O}{$v^1_{12, 2}$}
    \psfrag{P}{$v^2_{12, 2}$}
    \psfrag{Q}{$v^\Gamma_{12, 2}$}
    \psfrag{R}{$v^1_{12, 1}$}
    \psfrag{U}{$v^2_{12, 1}$}
    \psfrag{V}{$v^\Gamma_{12, 1}$}
    \psfrag{W}{$v^1_{\infty}$}
    \psfrag{X}{$v^\Gamma_{\infty}$}
    \psfrag{S}[c]{$s$}
    \psfrag{T}[c]{$t$}
    \psfrag{a}[c]{$h_{-12}^1$}
    \psfrag{b}[c]{$h_{-10}^1$}
    \psfrag{c}[c]{$h_{-2}^1$}
    \psfrag{d}[c]{$h_{0}^1$}
    \psfrag{e}[c]{$h_{2}^1$}
    \psfrag{f}[c]{$h_{10}^1$}
    \psfrag{g}[c]{$h_{12}^1$}
    \psfrag{h}[c]{$h_{-12}^2$}
    \psfrag{i}[c]{$h_{-11}^2$}
    \psfrag{j}[c]{$h_{-10}^2$}
    \psfrag{k}[c]{$h_{-9}^2$}
    \psfrag{l}[c]{$h_{-2}^2$}
    \psfrag{m}[c]{$h_{-1}^2$}
    \psfrag{n}[c]{$h_{0}^2$}
    \psfrag{o}[c]{$h_{1}^2$}
    \psfrag{p}[c]{$h_{2}^2$}
    \psfrag{q}[c]{$h_{9}^2$}
    \psfrag{r}[c]{$h_{10}^2$}
    \psfrag{s}[c]{$h_{11}^2$}
    \psfrag{t}[l]{$h_{12}^2$}
    \psfrag{u}{$S_{9, 1}$}
    \psfrag{v}{$S_{7, 1}$}
    \psfrag{w}{$S_{-9, 5}$}
    \psfrag{x}{$S_{-9, 4}$}
    \psfrag{y}{$M^{\tau+1}(h^1_{-10}, h^1_{-8})$}
    \psfrag{z}{$M^{\tau+1}(h^2_{-10}, h^2_{-9})$}
    %
    %
    \includegraphics[width=0.9\linewidth]{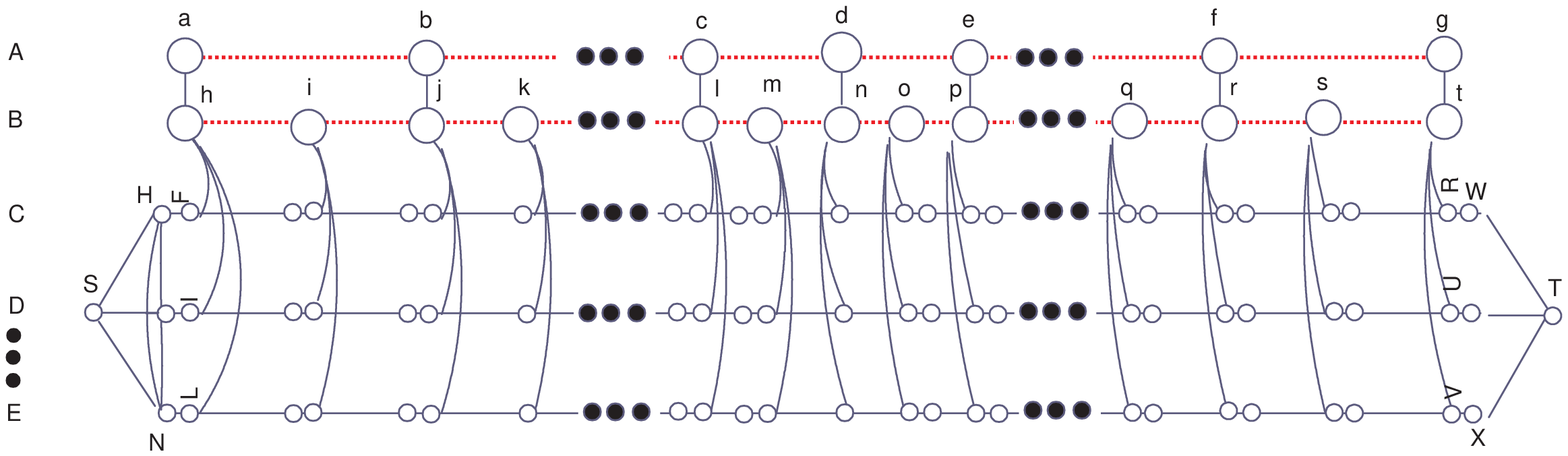}
    }
  \caption{\footnotesize An example of $\PRgraph$ where $\Lambda=2$ and $2\leq \kappa<3$. }\label{fig:graph_simpler}
\end{figure*}

Before we describe the construction of $\graph$, we first describe a network called $\PRgraph$ which is a slight modification of the network $F^K_m$ introduced in \cite{PelegR00}. In the next section, we show how we modify $\PRgraph$ to obtain $\graph$.

$\graph$ has three parameters, a real $\kappa\geq 1$ and two integers $\Gamma\geq 1$ and $\Lambda\geq 2$.\footnote{Note that we could restrict $\kappa$ to be an integer here since $F(\Gamma, \kappa, \Lambda)=F(\Gamma, \kappa', \Lambda)$ for any $\Lambda$, $\Gamma$, $\kappa$ and $\kappa'$ such that $\lfloor\kappa\rfloor=\lfloor\kappa'\rfloor$. However, we will need $\kappa$ to be a real when we define $\graph$ so we allow it to be a real here as well to avoid confusion.}
The two basic units in the construction of $\PRgraph$  are {\em highways} and {\em paths}.

\paragraph{Highways.} There are $\lfloor\kappa\rfloor$ highways, denoted by $\cH^1$, $\cH^2$, $\ldots$, $\cH^{\lfloor\kappa\rfloor}$. The highway $\cH^i$ is a path of $2\lceil\kappa\rceil\Lambda^i+1$ nodes, i.e.,
\begin{align*}
V(\cH^i) &= \{h_0^i, h_{\pm\Lambda^{\lfloor\kappa\rfloor-i}}^i, h_{\pm2\Lambda^{\lfloor\kappa\rfloor-i}}^i, h_{\pm3\Lambda^{\lfloor\kappa\rfloor-i}}^i, \dots, h_{\pm\lceil\kappa\rceil\Lambda^i\Lambda^{\lfloor\kappa\rfloor-i}}^i\}\\
E(\cH^i) &= \{(h^i_{-(j+1)\Lambda^{\lfloor\kappa\rfloor-i}}, h^i_{-j\Lambda^{\lfloor\kappa\rfloor-i}}), (h^i_{j\Lambda^{\lfloor\kappa\rfloor-i}}, h^i_{(j+1)\Lambda^{\lfloor\kappa\rfloor-i}}) \mid 0 \le j <\lceil\kappa\rceil\Lambda^i\}\,.
\end{align*}

We connect the highways by adding edges between nodes of the same subscripts, i.e., for any $0<i\leq \lfloor\kappa\rfloor$ and $-\lceil\kappa\rceil\Lambda^i \le j \le \lceil\kappa\rceil\Lambda^i$, there is an edge between $h^i_{j\Lambda^{\lfloor\kappa\rfloor-i}}$ and $h^{i+1}_{j\Lambda^{\lfloor\kappa\rfloor-i}}$.

For any $j\neq 0$, let 
\begin{align}
\phi_j=1 ~~\mbox{if $j=0$, and}~~ \phi'_j =\Lambda ~~\mbox{otherwise.}\label{eq:phi'_PelegR}
\end{align}

We use $\phi'_j$ to specify the number of nodes in the paths (defined next), i.e., each path will have $\sum_{j=-\lceil\kappa\rceil\Lambda^{\lfloor\kappa\rfloor}}^{\lceil\kappa\rceil\Lambda^{\lfloor\kappa\rfloor}} \phi'_j$ nodes. Note that
\begin{align}
\sum_{j=-\lceil\kappa\rceil\Lambda^{\lfloor\kappa\rfloor}}^{\lceil\kappa\rceil\Lambda^{\lfloor\kappa\rfloor}} \phi'_j
= (2\lceil\kappa\rceil\Lambda^{\lfloor\kappa\rfloor}+1) \Lambda
= \Theta(\kappa\Lambda^{\lfloor\kappa\rfloor+1}).\label{eq:sum_phi'_PelegR}
\end{align}

\paragraph{Paths.}
There are $\Gamma$ paths, denoted by $\cP^1, \cP^2, \ldots, \cP^{\Gamma}$. To construct each path, we first construct its subpaths as follows. For each node $h^{\lfloor\kappa\rfloor}_j$ in $\cH^{\lfloor\kappa\rfloor}$ and any $0< i\leq \Gamma$, we create a subpath of $\cP^i$, denoted by $\cP^i_j$, having $\phi'_j$ nodes. Denote nodes in $\cP^i_j$ in order by $v^i_{j, 1}, v^i_{j, 2}, \ldots, v^i_{j, \phi'_j}$. We connect these paths together to form $\cP^i_j$, i.e., for any $j\geq 0$, we create edges $(v^i_{j, \phi'_j}, v^i_{j+1, 1})$ and $(v^i_{-j, \phi'_{-j}}, v^i_{-(j+1), 1})$. Let
$$v^i_{-\infty}=v^i_{-\lceil\kappa\rceil\Lambda^{\lfloor\kappa\rfloor}, \phi'_{-\lceil\kappa\rceil\Lambda^{\lfloor\kappa\rfloor}}} ~~~\mbox{and}~~~ v^i_{\infty}=v^i_{\lceil\kappa\rceil\Lambda^{\lfloor\kappa\rfloor}, \phi'_{\lceil\kappa\rceil\Lambda^{\lfloor\kappa\rfloor}}}\,.$$
These two nodes can be thought of as the leftmost and rightmost nodes of path $\cP^i$. We connect the paths together by adding edges between the leftmost (rightmost, respectively) nodes in the paths, i.e., for any $i$ and $i'$, we add edges $(v^i_{-\infty}, v^{i'}_{-\infty})$ ($(v^i_{\infty}, v^{i'}_{\infty})$, respectively).

We connect the highways and paths by adding an edge from each node $h^{\lfloor\kappa\rfloor}_j$ to $v^i_{j, 1}$.
We also create nodes $s$ and $t$ and connect $s$ ($t$, respectively) to all nodes $v^i_{-\infty}$ ($v^i_\infty$, respectively).
See Fig.~\ref{fig:graph_simpler} for an example.

\subsection{Description of $\graph$}\label{subsec:graph_description}

We now modify $\PRgraph$ to obtain $\graph$. Again, $\graph$ has three parameters, a real $\kappa\geq 1$ and two integers $\Gamma\geq 1$ and $\Lambda\geq 2$. The two basic units in the construction of $\graph$ are {\em highways} and {\em paths}. The highways are defined in exactly the same way as before. The main modification is the definition of $\phi'$ (cf. Eq.~\eqref{eq:phi'_PelegR}) which affects the number of nodes in the subpaths $\cP^i_j$ of each path $\cP^i$.

\paragraph{Definition of $\phi'$.} First, for a technical reason in the proof of Theorem~\ref{thm:cc_to_distributed}, we need $\phi'_j$ to be small when $|j|$ is small. Thus, we define the following notation $\phi$. For any $j$, define 
\[\phi_j=\left\lfloor \frac{|j|}{\Lambda^{\lfloor\kappa\rfloor-1}}\right\rfloor + 1\,.\]
Note that $\phi_j$ can be viewed as the number of nodes in $\cH_1$ with subscripts between $0$ and $j$, i.e.,
%
%
\begin{equation*}
\phi_j=
\begin{cases}
|\{h^1_{j'} \mid 0\leq j'\leq j\}| &\mbox{if $j\geq 0$}\\
|\{h^1_{j'} \mid j\leq j'\leq 0\}| & \mbox{if $j<0$}\,.
\end{cases}
\end{equation*}
We now define $\phi'$ as follows. For any $j\geq 0$, let
%
%
\[\phi'_j = \phi'_{-j} = \min\left\{\phi_j, \max(1, \lceil\lceil\kappa\rceil\Lambda^\kappa\rceil-\sum_{j'>j} \phi_{j'})\right\}\,.\]
%
%
The reason we define $\phi'$ this way is that we use it to specify the number of nodes in the paths (as described in the previous section) and we want to be able to control this number precisely. In particular, while each path $\cP^i$ in $\PRgraph$ has $\Theta(\kappa\Lambda^{\lfloor\kappa\rfloor+1})$ nodes (cf. Eq.~\eqref{eq:sum_phi'_PelegR}), the number of nodes in each path in $\graph$ is 
\begin{align}
\sum_{j=-\lceil\kappa\rceil\Lambda^{\lfloor\kappa\rfloor}}^{\lceil\kappa\rceil\Lambda^{\lfloor\kappa\rfloor}} \phi'_j = \Theta(\kappa\Lambda^\kappa).\label{eq:sum_phi'}
\end{align}
We need this precision so that we can deal with any value of $\ell$ when we prove Theorem~\ref{thm:rw_lower_bound} in Section~\ref{sec:main_theorem}.


\begin{figure*}
  \centering
  \tiny
    {
    \psfrag{A}[c]{$\cH^1$}
    \psfrag{B}[c]{$\cH^2$}
    \psfrag{C}[c]{$\cP^1$}
    \psfrag{D}[c]{$\cP^2$}
    \psfrag{E}[c]{$\cP^\Gamma$}
    \psfrag{F}{$v^1_{-12, 1}$}
    \psfrag{G}{$v^1_{-12, 2}$}
    \psfrag{H}{$v^1_{-\infty}$}
    \psfrag{I}{$v^2_{-12, 1}$}
    \psfrag{J}{$v^2_{-12, 2}$}
    \psfrag{K}{$v^2_{-\infty}$}
    \psfrag{L}{$v^\Gamma_{-12, 1}$}
    \psfrag{M}{$v^\Gamma_{-12, 2}$}
    \psfrag{N}{$v^\Gamma_{-\infty}$}
    \psfrag{O}{$v^1_{12, 2}$}
    \psfrag{P}{$v^2_{12, 2}$}
    \psfrag{Q}{$v^\Gamma_{12, 2}$}
    \psfrag{R}{$v^1_{12, 1}$}
    \psfrag{U}{$v^2_{12, 1}$}
    \psfrag{V}{$v^\Gamma_{12, 1}$}
    \psfrag{W}{$v^1_{\infty}$}
    \psfrag{X}{$v^\Gamma_{\infty}$}
    \psfrag{S}[c]{$s$}
    \psfrag{T}[c]{$t$}
    \psfrag{a}[c]{$h_{-12}^1$}
    \psfrag{b}[c]{$h_{-10}^1$}
    \psfrag{c}[c]{$h_{-2}^1$}
    \psfrag{d}[c]{$h_{0}^1$}
    \psfrag{e}[c]{$h_{2}^1$}
    \psfrag{f}[c]{$h_{10}^1$}
    \psfrag{g}[c]{$h_{12}^1$}
    \psfrag{h}[c]{$h_{-12}^2$}
    \psfrag{i}[c]{$h_{-11}^2$}
    \psfrag{j}[c]{$h_{-10}^2$}
    \psfrag{k}[c]{$h_{-9}^2$}
    \psfrag{l}[c]{$h_{-2}^2$}
    \psfrag{m}[c]{$h_{-1}^2$}
    \psfrag{n}[c]{$h_{0}^2$}
    \psfrag{o}[c]{$h_{1}^2$}
    \psfrag{p}[c]{$h_{2}^2$}
    \psfrag{q}[c]{$h_{9}^2$}
    \psfrag{r}[c]{$h_{10}^2$}
    \psfrag{s}[c]{$h_{11}^2$}
    \psfrag{t}[c]{$h_{12}^2$}
    \psfrag{u}{$S_{9, 1}$}
    \psfrag{v}{$S_{7, 1}$}
    \psfrag{w}{$S_{-9, 5}$}
    \psfrag{x}{$S_{-9, 4}$}
    \psfrag{y}{$M^{\tau+1}(h^1_{-10}, h^1_{-8})$}
    \psfrag{z}{$M^{\tau+1}(h^2_{-10}, h^2_{-9})$}
    %
    %
    \includegraphics[width=\linewidth]{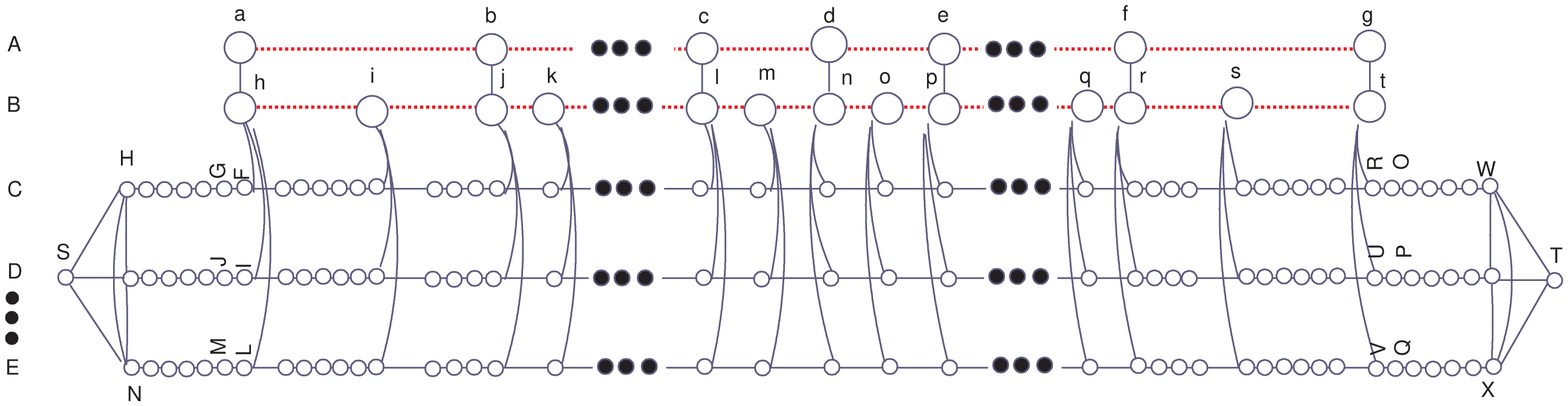}
    }
  \caption{\footnotesize An example of $\graph$ where $\kappa=2.5$ and $\Lambda=2$. The dashed edges (in red) have one copy while other edges have infinitely many copies. Note that $\phi'_{10}=4$ and thus there are $4$ nodes in each subpath $\cP^i_{10}$, for all $i$. Note also that $\phi'_{10}$ is less than $\phi_{10}$ which is $6$.}\label{fig:graph}
\end{figure*}

Finally, we make infinite copies of every edge except highway edges, i.e., those in $\cup_{i=1}^{\lfloor\kappa\rfloor} E(\cH^i)$. (In other words, we make them have infinite capacity). As mentioned earlier, we do this so that we will have a freedom to specify the number of copies later on when we prove Theorem~\ref{thm:rw_lower_bound} in Section~\ref{sec:main_theorem}. Observe that if Theorem~\ref{thm:cc_to_distributed} then it also holds when we set the numbers of edge copies in $\graph$ to some specific numbers.
Fig.~\ref{fig:graph} shows an example of $\graph$.

\subsection{Terminologies}

For any numbers $i$, $j$, $i'$, and $j'$, we say that $(i', j')\geq (i, j)$ if $i'>i$ or ($i'=i$ and $j'\geq j$).
For any $-\lceil\kappa\rceil\Lambda^{\lfloor\kappa\rfloor}\leq i\leq \lceil\kappa\rceil\Lambda^{\lfloor\kappa\rfloor}$ and $1\leq j\leq \phi'_i$, define the {\em $(i, j)$-set} as
\begin{equation*}
S_{i, j} =
\begin{cases}
\{h^x_{i'}\ |\ 1\leq x\leq \kappa,\ i'\leq i\}\cup\ \{v^x_{i', j'}\ |\  1\leq x\leq \Gamma,\ (i, j)\geq (i', j')\} \cup \{s\} &\text{if $i\geq 0$}\\
\{h^x_{i'}\ |\ 1\leq x\leq \kappa,\ i'\geq i\}\cup\ \{v^x_{i', j'}\ |\  1\leq x\leq \Gamma,\ (-i, j)\geq (-i', j')\} \cup \{r\} &\text{if $i<0$}\,.
\end{cases}
\end{equation*}
See Fig.~\ref{fig:simulation} for an example. For convenience, for any $i>0$, let 
$$S_{i, 0}=S_{i-1, \phi'_{i-1}}~~~~\mbox{and}~~~~S_{-i, 0}=S_{-(i-1), \phi'_{-(i-1)}}\,,$$ 
and, for any $j$, let 
$$S_{\lceil\kappa\rceil\Lambda^{\lfloor\kappa\rfloor}+1, j}=S_{\lceil\kappa\rceil\Lambda^{\lfloor\kappa\rfloor}, \phi'_{\lceil\kappa\rceil\Lambda^{\lfloor\kappa\rfloor}}}~~~~\mbox{and}~~~~S_{-\lceil\kappa\rceil\Lambda^{\lfloor\kappa\rfloor}-1, j}=S_{-\lceil\kappa\rceil\Lambda^{\lfloor\kappa\rfloor}, \phi'_{-\lceil\kappa\rceil\Lambda^{\lfloor\kappa\rfloor}}}\,.$$

Let $\mathcal{A}$ be any {\em deterministic} distributed algorithm run on $\graph$ for computing a function $f$. Fix any input strings $\bar{x}$ and $\bar{y}$ given to $s$ and $t$ respectively. Let $\varphi_\cA(\bar{x}, \bar{y})$ denote the execution of $\mathcal{A}$ on $\bar{x}$ and $\bar{y}$. Denote the {\em state} of the node $v$ at the end of time $\tau$ during the execution $\varphi_\cA(\bar{x}, \bar{y})$ by $\sigma_\cA(v, \tau, \bar{x}, \bar{y})$. Let $\sigma_\cA(v, 0, \bar{x}, \bar{y})$ be the state of the node $v$ before the execution $\varphi_\cA(\bar{x}, \bar{y})$ begins. Note that $\sigma_\cA(v, 0, \bar{x}, \bar{y})$ is independent of the input if $v\notin \{s, t\}$, depends only on $\bar{x}$ if $v=s$ and depends only on $\bar{y}$ if $v=t$. Moreover, in two different executions $\varphi_\cA(\bar{x}, \bar{y})$ and $\varphi_\cA(\bar{x}', \bar{y}')$, a node reaches the same state at time $\tau$ (i.e., $\sigma_\cA(v, \tau, \bar{x}, \bar{y})=\sigma_\cA(v, \tau, \bar{x}', \bar{y}')$) if and only if it receives the same sequence of messages on each of its incoming links.

For a given set of nodes $U=\{v_1, \ldots, v_\ell\}\subseteq V$, a {\em configuration}
%
\[C_\cA(U, \tau, \bar{x}, \bar{y}) = <\sigma_\cA(v_1, \tau, \bar{x}, \bar{y}), \ldots, \sigma_\cA(v_\ell, \tau, \bar{x}, \bar{y})>\]
%
%
is a vector of the states of the nodes of $U$ at the end of time $\tau$ of the execution $\varphi_\cA(\bar{x}, \bar{y})$.
From now on, to simplify notations, when $\cA$, $\bar{x}$ and $\bar{y}$ are clear from the context, we use $C^\tau_{i, j}$ to denote $C_\cA(S_{i, j}, \tau, \bar{x}, \bar{y})$.
%


\begin{figure*}
  \centering
  \tiny
    {
    \psfrag{A}[c]{$\cH^1$}
    \psfrag{B}[c]{$\cH^2$}
    \psfrag{C}[c]{$\cP^1$}
    \psfrag{D}[c]{$\cP^2$}
    \psfrag{E}[c]{$\cP^\Gamma$}
    \psfrag{F}{$v^1_{-11, 4}$}
    \psfrag{G}{$v^1_{-11, 5}$}
    \psfrag{H}{$S_{1, 1}$}
    \psfrag{I}{$v^2_{-11, 4}$}
    \psfrag{J}{$v^2_{-11, 5}$}
    \psfrag{K}{$v^2_{-\infty}$}
    \psfrag{L}{$v^\Gamma_{-11, 4}$}
    \psfrag{M}{$v^\Gamma_{-11, 5}$}
    \psfrag{N}{$v^\Gamma_{-\infty}$}
    \psfrag{O}{$v^1_{11, 1}$}
    \psfrag{P}{$v^2_{11, 1}$}
    \psfrag{Q}{$v^\Gamma_{11, 1}$}
    \psfrag{R}{$v^1_{9, 1}$}
    \psfrag{U}{$v^2_{9, 1}$}
    \psfrag{V}{$v^\Gamma_{9, 1}$}
    \psfrag{W}{$S_{-11, 0}=S_{-10, 4}$}
    \psfrag{X}{$v^\Gamma_{\infty}$}
    \psfrag{S}[c]{$s$}
    \psfrag{T}[c]{$t$}
    \psfrag{a}[c]{$h_{-12}^1$}
    \psfrag{b}[c]{$h_{-10}^1$}
    \psfrag{c}[c]{$h_{-2}^1$}
    \psfrag{d}[c]{$h_{0}^1$}
    \psfrag{e}[c]{$h_{2}^1$}
    \psfrag{f}[c]{$h_{10}^1$}
    \psfrag{g}[c]{$h_{12}^1$}
    \psfrag{h}[c]{$h_{-12}^2$}
    \psfrag{i}[c]{$h_{-11}^2$}
    \psfrag{j}[c]{$h_{-10}^2$}
    \psfrag{k}[c]{$h_{-9}^2$}
    \psfrag{l}[c]{$h_{-2}^2$}
    \psfrag{m}[c]{$h_{-1}^2$}
    \psfrag{n}[c]{$h_{0}^2$}
    \psfrag{o}[c]{$h_{1}^2$}
    \psfrag{p}[c]{$h_{2}^2$}
    \psfrag{q}[c]{$h_{9}^2$}
    \psfrag{r}[c]{$h_{10}^2$}
    \psfrag{s}[c]{$h_{11}^2$}
    \psfrag{t}[c]{$h_{12}^2$}
    \psfrag{u}{$S_{11, 1}$}
    \psfrag{v}{$S_{9, 1}$}
    \psfrag{w}{$S_{-11, 6}$}
    \psfrag{x}{$S_{-11, 5}$}
    \psfrag{y}{$M^{8}(h^1_{-12}, h^1_{-10})$}
    \psfrag{z}{$M^{8}(h^2_{-12}, h^2_{-11})$}
    \includegraphics[width=\linewidth]{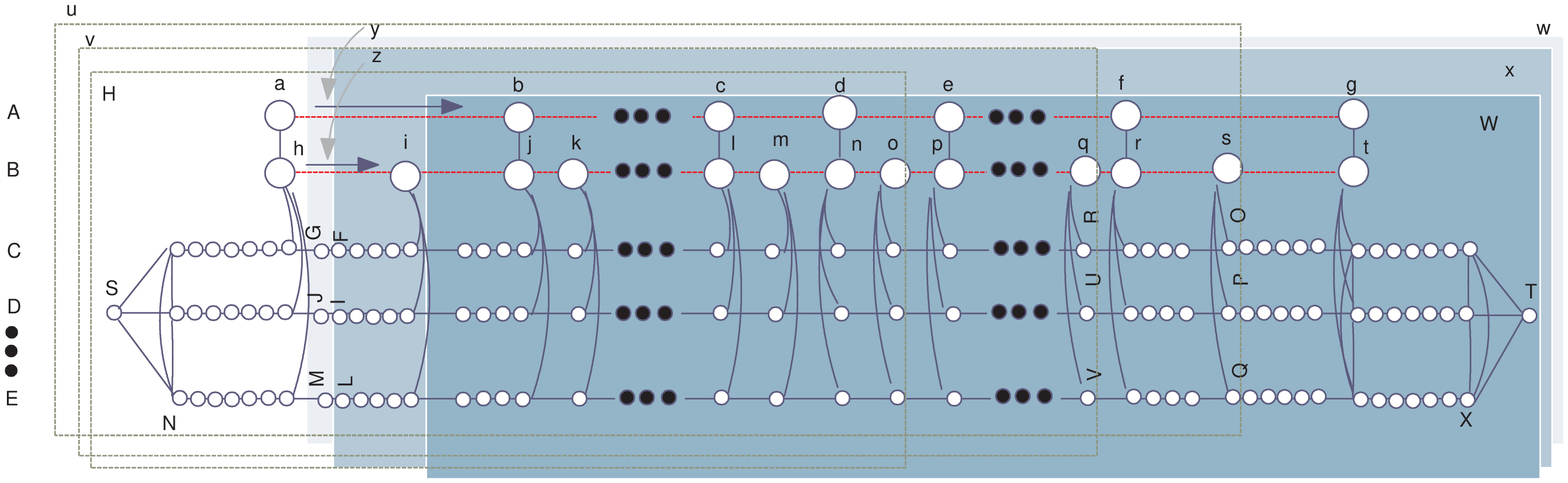}
    }
%
  \caption{\footnotesize An example of round $11$ in the proof of Theorem~\ref{thm:cc_to_distributed} (see detail in Example~\ref{ex:protocol}).} \label{fig:simulation}
\end{figure*}

\subsection{Proof of Theorem~\ref{thm:cc_to_distributed}}
%
%
Let $G=\graph$. Let $f$ be the function in the theorem statement. Let $\mathcal{A}_\epsilon$ be any $\epsilon$-error distributed algorithm for computing $f$ on $G$. Fix a random string $\bar{r}$ used by $\mathcal{A}_\epsilon$ (shared by all nodes in $G$) and consider the {\em deterministic} algorithm $\mathcal{A}$ run on the input of $\mathcal{A}_\epsilon$ and the fixed random string $\bar{r}$.
Let $T_{\mathcal{A}}$ be the worst case running time of algorithm $\mathcal{A}$ (over all inputs).
We only consider $T_\mathcal{A}\leq \kappa\Lambda^{\kappa}$, as assumed in the theorem statement.
We show that Alice and Bob, when given $\bar{r}$ as the public random string, can simulate $\mathcal{A}$ using $(2\kappa \log{n})T_\mathcal{A}$ communication bits in $8T_\mathcal{A}/(\kappa\Lambda)$ rounds, as follows. (We provide an example in the end of this section.)


\paragraph{Rounds, Phases, and Iterations.}
For convenience, we will name the rounds backward, i.e., Alice and Bob start at round $\lceil\kappa\rceil\Lambda^{\lfloor\kappa\rfloor}$ and proceed to round $\lceil\kappa\rceil\Lambda^{\lfloor\kappa\rfloor}-1$, $\lceil\kappa\rceil\Lambda^{\lfloor\kappa\rfloor}-2$, and so on.
Each round is divided into two {\em phases}, i.e., when Alice sends messages and Bob sends messages (recall that Alice sends messages first in each iteration). Each phase of round $r$ is divided into $\phi'_r$ {\em iterations}. Each iteration simulates one round of algorithm $\cA$.
We call the $i^{th}$ iteration of round $r$ when Alice (Bob, respectively) sends messages the {\em iteration $I_{r, A, i}$} ($I_{r, B, i}$, respectively). Therefore, in each round $r$ we have the following order of iterations: $I_{r, A, 1}$, $I_{r, A, 2}$, $\ldots$, $I_{r, A, \phi'_r}$, $I_{r, B, 1}$, $\ldots$, $I_{r, B, \phi'_r}$. For convenience, we refer to the time before communication begins as round $\lceil\kappa\rceil\Lambda^{\lfloor\kappa\rfloor}+1$ and let $I_{r, A, 0}=I_{r+1, A, \phi'_{r+1}}$ and $I_{r, B, 0}=I_{r+1, B, \phi'_{r+1}}$.

Our goal is to simulate one round of algorithm $\cA$ per iteration. That is, after iteration $I_{r, B, i}$ finishes, we will finish the $(\sum_{r'=r+1}^{\lceil\kappa\rceil\Lambda^{\lfloor\kappa\rfloor}} \phi'_{r'}+i)^{th}$ round of algorithm $\cA$. Specifically, we let
$$t_r=\sum_{r'=r+1}^{\lceil\kappa\rceil\Lambda^{\lfloor\kappa\rfloor}} \phi'_{r'}$$
and our goal is to construct a protocol with properties as in the following lemma.

\begin{lemma}\label{lem:after_iteration} There exists a protocol such that there are at most $\kappa\log n$ bits sent in each iteration and satisfies the following properties. For any $r\geq 0$ and $0\leq i\leq \phi'_r$,
\begin{enumerate}
\item after $I_{r, A, i}$ finishes, Alice and Bob know $C^{t_r+i}_{r-i\Lambda^{\lfloor\kappa\rfloor-1}, 1}$ and $C^{t_r+i}_{-r, \phi'_{-r}-i}$, respectively, and \label{property:alice}
\item after $I_{r, B, i}$ finishes, Alice and Bob know $C^{t_r+i}_{r, \phi'_{r}-i}$ and $C^{t_r+i}_{-r+i\Lambda^{\lfloor\kappa\rfloor-1}, 1}$, respectively. \label{property:bob}
\end{enumerate}
\end{lemma}

\begin{proof}
We first argue that the properties hold for iteration $I_{\lceil\kappa\rceil\Lambda^{\lfloor\kappa\rfloor}+1, A, 0}$, i.e., before Alice and Bob starts communicating. After round $r=\lceil\kappa\rceil\Lambda^{\lfloor\kappa\rfloor}$ starts, Alice can compute $C^0_{r+1, 0}=C^0_{r+1, 1}=C^0_{r, \phi'_{r}}$ which contains the states of all nodes in $\graph$ except $t$. She can do this because every node except $s$ and $t$ has the same state regardless of the input and the state of $s$ depends only on her input string $\bar{x}$.  Similarly, Bob can compute $C^0_{-(r+1), 0}=C^0_{-(r+1), 1}=C^0_{r, \phi'_{r}}$  which depends only on his input $\bar{y}$.

Now we show that, if the lemma holds for any iteration $I_{r, A, i-1}$ then it also holds for iteration $I_{r, A, i}$ as well. Specifically, we show that if Alice and Bob know $C^{t_r+i-1}_{r-(i-1)\Lambda^{\lfloor\kappa\rfloor-1}, 1}$ and $C^{t_r+i-1}_{-r, \phi'_{-r}-(i-1)}$, respectively, then they will know $C^{t_r+i}_{r-i\Lambda^{\lfloor\kappa\rfloor-1}, 1}$ and $C^{t_r+i}_{-r, \phi'_{-r}-i}$, respectively, after Alice sends at most $\kappa\log n$ messages.

First we show that Alice can compute $C^{t_r+i}_{r-i\Lambda^{\lfloor\kappa\rfloor-1}, 1}$ without receiving any message from Bob. Recall that Alice can compute $C^{t_r+i}_{r-i\Lambda^{\lfloor\kappa\rfloor-1}, 1}$ if she knows
\squishlist
\item $C^{t_r+i-1}_{r-i\Lambda^{\lfloor\kappa\rfloor-1}, 1}$, and
\item all messages sent to all nodes in $S_{r-i\Lambda^{\lfloor\kappa\rfloor-1}, 1}$ at time $t_r+i$ of algorithm $\cA$.
\squishend
By assumption, Alice knows $C^{t_r+i-1}_{r-(i-1)\Lambda^{\lfloor\kappa\rfloor-1}, 1}$ which implies that she knows $C^{t_r+i-1}_{r-i\Lambda^{\lfloor\kappa\rfloor-1}, 1}$ since
$$S_{r-i\Lambda^{\lfloor\kappa\rfloor-1}, 1} \subseteq S_{r-(i-1)\Lambda^{\lfloor\kappa\rfloor-1}, 1}\,.$$
Moreover, observe that all neighbors of all nodes in  $S_{r-i\Lambda^{\lfloor\kappa\rfloor-1},1}$ are in $S_{r-(i-1)\Lambda^{\lfloor\kappa\rfloor-1}, 1}$. Thus, Alice can compute all messages sent to all nodes in $S_{r-i\Lambda^{\lfloor\kappa\rfloor-1}, 1}$ at time $t_r+i$ of algorithm $\cA$. Therefore, Alice can compute $C^{t_r+i}_{r+i\Lambda^{\lfloor\kappa\rfloor-1}, 1}$ without receiving any message from Bob.

Now we show that Bob can compute $C^{t_r+i}_{-r, \phi'_{-r}-i}$ by receiving at most $\kappa\log n$ bits from Alice and use the knowledge of $C^{t_r+i-1}_{-r, \phi'_{-r}-i+1}$. Note that Bob can compute $C^{t_r+i}_{-r, \phi'_{-r}-i}$ if he knows
\squishlist
\item $C^{t_r+i-1}_{-r, \phi'_{-r}-i}$, and
\item all messages sent to all nodes in $S_{-r, \phi'_{-r}-i}$ at time $t_r+i$ of algorithm $\cA$.
\squishend
By assumption, Bob knows $C^{t_r+i-1}_{-r, \phi'_{-r}-i+1}$ which implies that he knows $C^{t_r+i-1}_{-r, \phi'_{-r}-i}$ since $S_{-r, \phi'_{-r}-i}\subseteq S_{-r, \phi'_{-r}-i+1}$.
Moreover, observe that all neighbors of all nodes in $S_{-r, \phi'_{-r}-i}$ are in $S_{-r, \phi'_{-r}-i+1}$, except
\begin{align*}
&h^{\lfloor\kappa\rfloor}_{-(r+1)}, h^{\lfloor\kappa\rfloor-1}_{-(\lfloor r/\Lambda\rfloor+1)}, \ldots, h^{\lfloor\kappa\rfloor-i}_{-(\lfloor r/\Lambda^i\rfloor+1)}, \ldots, h^1_{-(\lfloor r/\Lambda^{\lfloor\kappa\rfloor-1}\rfloor+1)}.
\end{align*}
%
%
In other words, Bob can compute all messages sent to all nodes in $S_{-r, \phi'_{-r}-i}$ at time $t_r+i$ except
\begin{align*}
&M^{t_r+i}(h^{\lfloor\kappa\rfloor}_{-(r+1)}, h^{\lfloor\kappa\rfloor}_{-r}), \ldots, M^{t_r+i}(h^{\lfloor\kappa\rfloor-i}_{-(\lfloor r/\Lambda^i\rfloor+1)}, h^{\lfloor\kappa\rfloor-i}_{-\lfloor r/\Lambda^i\rfloor}),
\ldots, M^{t_r+i}(h^1_{-(\lfloor r/\Lambda^{\lfloor\kappa\rfloor-1}\rfloor+1)}, h^1_{-(\lfloor r/\Lambda^{\lfloor\kappa\rfloor-1}\rfloor})
\end{align*}
%
%
where $M^{t_r+i}(u, v)$ is the message sent from $u$ to $v$ at time $t_r+i$ of algorithm $\cA$. Observe further that Alice can compute these messages because she knows $C^{t_r+i-1}_{r-(i-1)\Lambda^{\lfloor\kappa\rfloor-1}, 1}$ which contains the states of
\begin{align*}
h^{\lfloor\kappa\rfloor}_{-(r+1)}, \ldots, h^{\lfloor\kappa\rfloor-i}_{-(\lfloor r/\Lambda^i\rfloor+1)}, \ldots,~~~h^1_{-(\lfloor r/\Lambda^{\lfloor\kappa\rfloor-1}\rfloor+1)}
\end{align*}
%
%
at time $t_r+i-1$. (In particular, $C^{t_r+i-1}_{r-(i-1)\Lambda^{\lfloor\kappa\rfloor-1}, 1}$ is a superset of $C^{t_r+i-1}_{0, 1}$ which contains the states of $h^{\lfloor\kappa\rfloor}_{-(r+1)}$, $\ldots$, $h^1_{-(\lfloor r/\Lambda^{\lfloor\kappa\rfloor-1}\rfloor+1)}$.) So, Alice can send these messages to Bob and Bob can compute $C^{t_r+i}_{-r, \phi'_{-r}-i}$ at the end of the iteration. Each of these messages contains at most $\log n$ bits since each of them corresponds to a message sent on one edge. Therefore, Alice sends at most $\kappa \log n$ bits to Bob in total. This shows the first property.

After Alice finishes sending messages, the two parties will switch their roles and a similar protocol can be used to show that the second property, i.e., if the lemma holds for any iteration $I_{r, B, i-1}$ then it also holds for iteration $I_{r, B, i}$ as well. That is, if Alice and Bob know $C^{t_r+i-1}_{r, \phi'_{r}-(i-1)}$ and $C^{t_r+i-1}_{-r+(i-1)\Lambda^{\lfloor\kappa\rfloor-1}, 1}$, respectively, then Bob can send $\kappa \log n$ bits to Alice so that they can compute $C^{t_r+i}_{r, \phi'_{r}-i}$ and $C^{t_r+i}_{-r+i\Lambda^{\lfloor\kappa\rfloor-1}, 1}$, respectively.
\end{proof}




Let $P$ be the protocol as in Lemma~\ref{lem:after_iteration}. Alice and Bob will run protocol $P$ until round $r'$, where $r'$ is the largest number such that $t_{r'}+\phi'_{r'}\geq T_\cA$.
Lemma~\ref{lem:after_iteration} implies that after iteration $I_{r', B, T_\cA-t_{r'}}$, Bob knows $$C^{t_{-r'}+T_\cA-t_{r'}}_{-r', \phi'_{-r'}-T_\cA+t_{r'}}=C^{T_\cA}_{-r', \phi'_{-r'}-T_\cA+t_{r'}}$$
(note that $\phi'_{-r'}-T_\cA+t_{r'}\geq 0$).
%
%
In particular, Bob knows the state of node $t$ at time $T_\cA$, i.e., he knows $\sigma_\cA(t, T_\cA, \bar{x}, \bar{y})$. Thus, Bob can output the output of $\cA$ which is output from $t$.

Since $\mathcal{A}_\epsilon$ is $\epsilon$-error, the probability (over all possible shared random strings) that $\mathcal{A}$ outputs the correct value of $f(\bar{x}, \bar{y})$ is at least $1-\epsilon$. Therefore, the communication protocol run by Alice and Bob is $\epsilon$-error as well. The number of rounds is bounded as in the following claim.

\begin{claim}
If algorithm $\cA$ finishes in time $T_\cA\leq \lceil\kappa\rceil\Lambda^\kappa$ then $r'>\lceil\kappa\rceil\Lambda^\kappa-8T_\cA/(\lceil\kappa\rceil\Lambda)$. In other words, the number of rounds Alice and Bob need to simulate $\cA$ is $8T_\cA/(\lceil\kappa\rceil\Lambda)$
\end{claim}
\begin{proof}
Let $R^*=8T_\cA/(\lceil\kappa\rceil\Lambda)$ and let $r^*=\Lambda^{\lfloor\kappa\rfloor}-R^*+1$. Assume for the sake of contradiction that Alice and Bob need more than $R^*$ rounds. This means that $r'<r^*$.
%
%
%
%
Alice and Bob requiring more than $R^*$ rounds implies that 
%
\begin{align}
\sum_{r=r^*}^{\lceil\kappa\rceil\Lambda^{\lfloor\kappa\rfloor}} \phi'_r = t_{r^*}+\phi'_{r^*} < T_\cA \leq \lceil\kappa\rceil\Lambda^\kappa\,.\label{eq:round_bound}
\end{align}
It follows that for any $r\geq r^*$,
\begin{align}
\phi'_r &= \min\left(\phi(h^{k'}_{r}), \max(1, \lceil\lceil\kappa\rceil\Lambda^\kappa\rceil-\sum_{r'>r} \phi_{r'})\right) \label{eq:a1}\\
&= \phi_r \label{eq:a2}\\
&= \left\lfloor \frac{r}{\Lambda^{\lfloor\kappa\rfloor-1}}\right\rfloor +1 \label{eq:a3}
\end{align}
where Eq.~\eqref{eq:a1} follows from the definition of $\phi'_r$, Eq.~\eqref{eq:a2} is because $\sum_{r\geq r^*} \phi'_r< \lceil\kappa\rceil\Lambda^\kappa$, and Eq.~\eqref{eq:a3} is by the definition of $\phi_r$.
%
%
%
Therefore, the total number of steps that can be simulated by Alice and Bob up to round $r^*$ is
\begin{align*}
\sum_{r=r^*}^{\lceil\kappa\rceil\Lambda^{\lfloor\kappa\rfloor}} \phi'_r & = \sum_{r=r^*}^{\lceil\kappa\rceil\Lambda^{\lfloor\kappa\rfloor}} \left(\left\lfloor \frac{r}{\Lambda^{\lfloor\kappa\rfloor-1}}\right\rfloor +1\right)\\
&\geq \Lambda^{\lfloor\kappa\rfloor-1} \sum_{i=1}^{\lfloor R^*/\Lambda^{\lfloor\kappa\rfloor-1}\rfloor} (\lceil\kappa\rceil\Lambda-i)\\
&\geq \Lambda^{\lfloor\kappa\rfloor-1} \cdot \frac{\lfloor R^*/\Lambda^{\lfloor\kappa\rfloor-1}\rfloor (\lceil\kappa\rceil\Lambda-1)}{2}\\
&\geq \frac{R^*\lceil\kappa\rceil\Lambda}{8}\\
&\geq T_\cA
\end{align*}
contradicting Eq.~\eqref{eq:round_bound}.
\end{proof}

%

Since there are at most $\kappa \log n$ bits sent in each iteration and Alice and Bob runs $P$ for $T_\cA$ iterations, the total number of bits exchanged is at most $(2\kappa\log n)T_\mathcal{A}$. This completes the proof of Theorem~\ref{thm:cc_to_distributed}.

\begin{example}\label{ex:protocol}
Fig.~\ref{fig:simulation} shows an example of the protocol we use above. Before iteration $I_{11, A, 1}$ begins, Alice and Bob know $C^{7}_{11, 1}$ and $C^7_{-11, 5}$, respectively (since Alice and Bob already simulated $\cA$ for $\phi'_{12}=7$ steps in round $12$). Then, Alice computes and sends $M^{8}(h^2_{-12}, h^2_{-11})$ and $M^{8}(h^1_{-12}, h^1_{-10})$ to Bob.  Alice and Bob then compute $C^{8}_{11, 1}$ and $C^{8}_{-11, 6}$, respectively, at the end of iteration $I_{11, A, 1}$. After they repeat this process for five more times, i.e. Alice sends
$$M^{9}(h^2_{-12}, h^2_{-11}), M^{10}(h^2_{-12}, h^2_{-11}), \ldots, M^{13}(h^2_{-12}, h^2_{-11}), ~~~\mbox{and}~~~$$
$$M^{9}(h^1_{-12}, h^1_{-10}), M^{10}(h^1_{-12}, h^1_{-10}), \ldots, M^{13}(h^1_{-12}, h^1_{-10})\,,$$
Bob will be able to compute $C^{13}_{-11, 0}=C^{13}_{-10, 4}$. Note that Alice is able to compute $C^8_{9, 1}$, $C^9_{7, 1}$, $\ldots$, $C^{12}_{1, 1}$ without receiving any messages from Bob so she can compute and send the previously mentioned messages to Bob.
\end{example}

\section{The pointer chasing problem} \label{sec:pointer_chasing}


In this section, we define the pointer chasing problem and prove its lower bound (Lemma~\ref{lem:PC_dist_lowerbound}) which will be used to prove Theorem~\ref{thm:rw_lower_bound} in the next section.

Informally, the $r$-round pointer chasing problem has parameters $r$ and $m$ and there are two players, which could be Alice and Bob or nodes $s$ and $t$, who receive functions $f_A:[m]\rightarrow[m]$ and $f_B:[m]\rightarrow[m]$, respectively. The goal is to compute a function starting from $1$ and alternatively applying $f_A$ and $f_B$ for $r$ times each, i.e., compute $f_B(\ldots f_A(f_B(f_A)))$ where $f_A$ and $f_B$ appear $r$ times each.
To be precise, let $\cF_m$ be the set of functions $f:[m]\rightarrow [m]$. For any $i\geq 0$ define $g^{i}: \cF_m\times\cF_m\rightarrow [m]$ inductively as
\[g^{0}(f_A, f_B)=1  ~~~\mbox{and}~~~\]
\[g^{i}(f_A, f_B)\begin{cases}
f_A(g^{i-1}(f_A, f_B)) & \mbox{if $i>0$ and $i$ is odd,}\\
f_B(g^{i-1}(f_A, f_B)) & \mbox{if $i>0$ and $i$ is even.}
\end{cases}\]
Also define function $\PC^{i, m}(f_A, f_B)=g^{2i}(f_A, f_B)$. The goal of the $r$-round pointer chasing problem is to compute $\PC^{r, m}(f_A, f_B)$.


Observe that if Alice and Bob can communicate for $r$ rounds then they can compute $\PC^{r, m}$ naively by exchanging $O(r\log m)$ bits. Interestingly, Nisan and Wigderson~\cite{NisanW93} show that if Alice and Bob are allowed only $r-1$ rounds then they essentially cannot do anything better than having Alice sent everything she knows to Bob.\footnote{In fact this holds even when Alice and Bob are allowed $r$ rounds but Alice cannot send a message in the first round.}


%

\begin{theorem}\cite{NisanW93}\label{thm:pointer_chasing}
$R^{(r-1)-cc-pub}_{1/3}(\PC^{r, m})=\Omega(m/r^{2}-r\log m)$.
\end{theorem}

\paragraph{The pointer chasing problem on $\graph$.} We now consider the pointer chasing problem on network $\graph$ where $s$ and $t$ receive $f_A$ and $f_B$ respectively. The following lemma follows from Theorem~\ref{thm:cc_to_distributed} and \ref{thm:pointer_chasing}.
%
%
\begin{lemma}\label{lem:PC_dist_lowerbound}
For any $\kappa$, $\Gamma$, $\Lambda\geq 2$, $m\geq \kappa^2\Lambda^{4\kappa}\log n$, $16\Lambda^{\kappa-1}\geq r>8\Lambda^{\kappa-1}$, $R^{\graph, s, t}_{1/3}(\PC^{r, m})=\Omega(\kappa\Lambda^{\kappa})$.
\end{lemma}

\begin{proof}
%
Let $r=R_{1/3}^{\graph, s, t}(\PC^{r, m})$. If $r> \kappa\Lambda^\kappa$ then we are done so we assume that $r \leq \kappa\Lambda^\kappa$. Thus,
\begin{align}
r&\geq \frac{R_{1/3}^{\frac{8R_{1/3}^{\graph, s, t}(\PC^{r, m})}{\kappa\Lambda}-cc-pub}(\PC^{r, m})}{(2\kappa\log n)} \label{eq:a}\\
&\geq \frac{R_{1/3}^{\frac{8\kappa\Lambda^\kappa}{\kappa\Lambda}-cc-pub}(\PC^{r, m})}{(2\kappa\log n)}\label{eq:b} \\
&= \Omega(\frac{(m(8\Lambda^{\kappa-1})^{-2}-8\Lambda^{\kappa-1}\log m)}{(\kappa\log n)}) \label{eq:c} \\
&= \Omega(\kappa\Lambda^{\kappa}) \label{eq:d}
\end{align}
where Eq.~\eqref{eq:a} is by Theorem~\ref{thm:cc_to_distributed} and the fact that $r\leq \kappa\Lambda^\kappa$, Eq.~\eqref{eq:b} uses the fact that the communication does not increase when we allow more rounds and $R_{1/3}^{\graph, s, t}(\PC^{r, m})\leq \kappa\Lambda^\kappa$, Eq.~\eqref{eq:c} follows from Theorem~\ref{thm:pointer_chasing} with the fact that $16\Lambda^{\kappa-1}\geq r>8\Lambda^{\kappa-1}$ and Eq.~\eqref{eq:d} is because $m\geq \kappa^2\Lambda^{4\kappa}\log n$.
%
%
%
\end{proof}

\section{Proof of the main theorem}\label{sec:main_theorem}
In this section, we prove Theorem~\ref{thm:rw_lower_bound}. An $\Omega(D)$ lower bound has already been shown (and is fairly straightforward) in \cite{DasSarmaNP09}; so we focus on showing the $\Omega(\sqrt{\ell D})$ lower bound. Moreover, we will prove the theorem only for the version where destination outputs source. This is because we can convert algorithms for the other two version to solve this version by adding $O(D)$ rounds. To see this, observe that once the source outputs the ID of the destination, we can take additional $O(D)$ rounds to send the ID of the source to the destination. Similarly, if nodes know their positions, the node with position $\ell$ can output the source's ID by taking additional $O(D)$ rounds to request for the source's ID. Theorem~\ref{thm:rw_lower_bound}, for the case where destination outputs source, follows from the following lemma.

\begin{lemma}
For any real $\kappa\geq 1$ and integers $\Lambda\geq 2$, and $\Gamma\geq 32\kappa^2\Lambda^{6\kappa-1}\log n$, there exists a family of networks $\cH$ such that any network $H\in \cH$ has $\Theta(\kappa\Gamma\Lambda^\kappa)$ nodes and diameter $D=\Theta(\kappa\Lambda)$, and
%
any algorithm for computing the destination of a random walk of length $\ell=\Theta(\Lambda^{2\kappa-1})$ requires $\Omega(\sqrt{\ell D})$ time on some network $H\in \cH$.
\end{lemma}
\begin{proof}
We show how to compute $\PC^{r, m}$ on $G=\graph$ by reducing the problem to the problem of sampling a random walk destination in some network $H_{f_A, f_B}$, obtained by restrict the number of copies of some edges in $G$, depending on input functions $f_A$ and $f_B$. We let $\cH$ be the family of network $H_{f_A, f_B}$ over all input functions. Note that for any input functions, an algorithm on $H_{f_A, f_B}$ can be run on $G$ with the same running time since every edge in $G$ has more capacity than its counterpart in $H_{f_A, f_B}$.
%

Let $r=16\Lambda^{\kappa-1}$ and $m=\kappa^2\Lambda^{5\kappa}\log n$. Note that $2rm\leq \Gamma$.
For any $i\leq r$ and $j\leq m$, let 
\[S^{i, j}=\cP^{2(i-1)m+j}~~~\mbox{and}~~~T^{i, j}=\cP^{2(i-1)m+m+j}\,.\]
That is, $S^{1, 1}=\cP^1$, $\ldots$, $S^{1, m}=\cP^m$, $T^{1, 1}=\cP^{m+1}$, $\ldots$, $T^{1, m}=\cP^{2m}$, $S^{2, 1}=\cP^{2m+1}$, $\ldots$, $T^{r, m}=\cP^{2rm}$.
%
%
Let $L$ be the number of nodes in each path. Note that $L=\Theta(\kappa\Lambda^\kappa)$ by Lemma~\ref{lem:graphsize}. Denote the nodes in $S^{i, j}$ from {\em left to right} by $s^{i, j}_1, \ldots, s^{i, j}_L$. (Thus, $s^{i, j}_1=v^{2(i-1)m+j}_{-\infty}$ and $s^{i,j}_L=v^{2(i-1)m+j}_{\infty}$.) Also denote the nodes in $T^{i, j}$ from {\em right to left} by $t^{i, j}_1, \ldots, t^{i, j}_L$. (Thus, $t^{i, j}_1=v^{2(i-1)m+m+j}_{\infty}$ and $t^{i,j}_L=v^{2(i-1)m+m+j}_{-\infty}$.) Note that for any $i$ and $j$, $s^{i, j}_1$ and $t^{i, j}_L$ are adjacent to $s$ while  $s^{i, j}_L$ and $t^{i, j}_1$ are adjacent to $t$.

Now we construct $H_{f_A, f_B}$. For simplicity, we fix input functions $f_A$ and $f_B$ and denote $H_{f_A, f_B}$ simply by $H$. To get $H$ we let every edge in $G$ have one copy (thus with capacity $O(\log n)$), except the following edges.
For any $i\leq r$, $j\leq m$, and $x<L$, we have $(6\Gamma\ell)^{2(i-1)L+x}$ copies of edges between nodes $s^{i,j}_x$ and $s^{i,j}_{x+1}$  and $(6\Gamma\ell)^{2(i-1)L+L+x}$ copies of edges between nodes $t^{i,j}_x$ and $t^{i,j}_{x+1}$.
Note that these numbers of copies of edges are always the same, regardless of the input $f_A$ and $f_B$.

Additionally, we have the following numbers of edges which depend on the input functions. First, $s$ specifies the following number of edges between its neighbors. For any $i\leq r$, $j\leq m$, we have $(6\Gamma\ell)^{2(i-1)L+L}$ copies of edges between nodes $t^{i,j}_L$ and $s^{i, f_A(j)}_{1}$. These numbers of edges can be specified in one round since both $s^{i,j}_1$ and $t^{i, f_A(j)}_{L}$ are adjacent to $s$. Similarly, we have $(6\Gamma\ell)^{2(i-1)L+2L}$ copies of edges between nodes $t^{i,j}_1$ and $s^{i+1, f_B(j)}_{L}$ which can be done in one round since both nodes are adjacent to $t$. This completes the description of $H$.

Now we use any random walk algorithm to compute the destination of a walk of length $\ell=2rL-1=\Theta(\Lambda^{2\kappa-1})$ on $H$ by starting a random walk at $s^{1, f(A)}_1$. If the random walk destination is $t^{r, j}_L$ for some $j$, then node $t$ outputs the number $j$; otherwise, node $t$ outputs an arbitrary number.

Now observe the following claim. 

\begin{claim}\label{claim:highprob}
Node $t$ outputs $\PC^{r, m}(f_A, f_B)$ with probability at least $2/3$.
\end{claim}
\begin{proof}
Let $P^*$ be the path consisting of nodes $s^{1, f_A(1)}_1$, $\ldots$, $s^{1, f_A(1)}_L$, $t^{1, f_B(f_A(1))}_1$, $\ldots$, $t^{1, f_B(f_A(1))}_L$, $s^{1, f_A(f_B(f_A(1)))}_1$, $\ldots$, $s^{i, g^{2i-1}(f_A, f_B)}_L$, $t^{i, g^{2i}(f_A, f_B)}_1$, $\ldots$, $t^{r, g^{2r}(f_A, f_B)}_L$. We claim that the random walk will follow path $P^*$ with probability at least $2/3$. The node of distance $(2rL-1)$ from $s^{1, f_A(1)}_1$ in this path is $t^{r, g^{2r}(f_A, f_B)}_L=t^{r, \PC^{r, m}(1)}_L$ and thus the algorithm described above will output $\PC^{r, m}(1)$ with probability at least $2/3$.

To prove the above claim, consider any node $u$ in path $P^*$. Let $u'$ and $u''$ be the node before and after $u$ in $P^*$, respectively. Let $m'$ and $m''$ be the number of multiedges $(u, u')$ and $(u, u'')$, respectively. Observe that $m''\geq 6\Gamma \ell m'$. Moreover, observe that there are at most $\Gamma$ edges between $u$ and other nodes. Thus, if a random walk is at $u$, it will continue to $u''$ with probability at least $1-\frac{1}{3\ell}$. By union bound, the probability that a random walk will follow $P^*$ is at least $1-\frac{1}{3}$, as claimed.
\end{proof}

Thus, if there is any random walk algorithm with running time $O(T)$ on all networks in $\cH$ then we can use such algorithm to solve $\PC^{r, m}$ (with error probability $1/3$) in time $O(T)$. Using the lower bound of computing solving $\PC^{r, m}$ in Lemma~\ref{lem:PC_dist_lowerbound}, the random walk computation also has a lower bound of $\Omega(\kappa\Lambda^\kappa)=\Omega(\sqrt{\ell D})$ as claimed.
\end{proof}

To prove Theorem~\ref{thm:rw_lower_bound} with the given parameters $n$, $D$ and $\ell$, we simply set $\Lambda$ and $\kappa$ so that $\kappa\Lambda=D$ and $\Lambda^{2\kappa-1}=\Theta(\ell)$. This choice of $\Lambda$ and $\kappa$ exists since $\ell\geq D$. Setting $\Gamma$ large enough so that $\Gamma\geq 32\kappa^2\Lambda^{6\kappa-1}\log n$ while $\Gamma=\Theta(n)$. (This choice of $\Gamma$ exists since $\ell\leq (n/(D^3\log n))^{1/4}$.) By applying the above lemma, Theorem~\ref{thm:rw_lower_bound} follows. \note{In this setting $n=\kappa\Gamma\Lambda^\kappa=\kappa(32\kappa^2\Lambda^{6\kappa-1}\log n)\Lambda^{\kappa}\leq 32D^3\ell^4\log n$} \danupon{Is it necessary to explain this with more detail?} 
\section{Conclusion}
In this paper we prove a tight unconditional lower bound on the time complexity of distributed random walk computation, implying that the algorithm in \cite{DasSarmaNPT-PODC10} is time optimal. To the best of our knowledge, this is the first lower bound that the diameter plays a role of multiplicative factor. Our proof technique comes from strengthening the connection between communication complexity and distributed algorithm lower bounds initially studied in \cite{DasSarmaHKKNPPW10} by associating {\em rounds} in communication complexity to the distributed algorithm running time, with network diameter as a trade-off factor.

There are many open problems left for random walk computation. One interesting open problem is showing a lower bound of performing a long walk. We conjecture that the same lower bound of $\tilde\Omega(\sqrt{\ell D})$ holds for any $\ell=O(n)$. However, it is not clear whether this will hold for longer walks. For example, one can generate a {\em random spanning tree} by computing a walk of length equals the cover time (using the version where every node knows their positions) which is $O(mD)$ where $m$ is the number of edges in the network (see detail in \cite{DasSarmaNPT-PODC10}). It is interesting to see if performing such a walk can be done faster.  Additionally, the upper and lower bounds of the problem of generating a random spanning tree itself is very interesting since its current upper bound of $\tilde O(\sqrt{m}D)$~\cite{DasSarmaNPT-PODC10} simply follows as an application of random walk computation~\cite{DasSarmaNPT-PODC10} while no lower bound is known.
Another interesting open problem prove the lower bound of $\tilde \Omega(\sqrt{K\ell D})$ for some value of $\ell$ for the problem of performing {\em $K$ walks} of length $\ell$.

In light of the success in proving distributed algorithm lower bounds from communication complexity in this and the previous work~\cite{DasSarmaHKKNPPW10}, it is also interesting to explore further applications of this technique. One interesting approach is to show a connection between distributed algorithm lower bounds and other models of communication complexity, such as multiparty and asymmetric communication complexity (see, e.g., \cite{KNbook}).
One particular interesting research topic is applying this technique to distance-related problems such as shortest $s$-$t$ path, single-source distance computation, and all-pairs shortest path. The lower bound of $\Omega(\sqrt{n})$ are shown in \cite{DasSarmaHKKNPPW10} for these types of problems. It is interesting to see if there is an $O(\sqrt{n})$-time algorithm for these problems (or any sub-linear time algorithm) or a time lower bound of $\omega(\sqrt{n})$ exists. 
The special cases of these problems on complete graphs (as noted in \cite{Elkin-sigact04}) are particularly interesting. 
Besides these problems, there are still some gaps between upper and lower bounds of problems considered in \cite{DasSarmaHKKNPPW10} such as the minimum cut and generalized Steiner forest.

\bibliographystyle{plain}
\bibliography{randomwalk-lowerbound}


\end{document}